\documentclass[letterpaper, 10 pt, conference]{ieeeconf}  

\IEEEoverridecommandlockouts                              

\usepackage{graphics} 
\usepackage{epsfig} 
\usepackage{mathptmx} 
\usepackage{times} 
\usepackage{amsmath} 
\usepackage{amssymb}  
\usepackage{amsmath}

\usepackage{amsthm}
\usepackage{amscd}
\usepackage{upgreek}
\usepackage{amsfonts}
\usepackage{graphicx}
\usepackage{fancyhdr}
\usepackage{anysize}
\usepackage{mathrsfs}
\usepackage[utf8]{inputenc}
\usepackage{mathtools}
\usepackage{cite}
\newtheoremstyle{theoremdd}
{0.75cm}
{0.6cm}
{\itshape}
{}
{\bfseries}
{}
{ }
{\thmname{#1}\thmnumber{ #2}\thmnote{ (#3)}:}
\theoremstyle{theoremdd}
\newtheorem{theorem}{Theorem}

\usepackage{scalerel,stackengine}
\stackMath
\newcommand\reallywidehat[1]{%
	\savestack{\tmpbox}{\stretchto{%
			\scaleto{%
				\scalerel*[\widthof{\ensuremath{#1}}]{\kern-.6pt\bigwedge\kern-.6pt}%
				{\rule[-\textheight/2]{1ex}{\textheight}}
			}{\textheight}%
		}{0.5ex}}%
	\stackon[1pt]{#1}{\tmpbox}%
}
\makeatletter
\def\endfigure{\end@float}
\makeatother
\title{\LARGE \bf
A Consensus Algorithm for Second-Order Systems Evolving on Lie Groups 
}

\author{ Akhil B Krishna$^{1}$, Farshad Khorrami$^{1,2}$ and Anthony Tzes$^{1,3}$
\thanks{$^{1}$Center for Artificial Intelligence and Robotics (CAIR), New York University Abu Dhabi (NYUAD), Abu Dhabi, United Arab Emirates (UAE).}
\thanks{$^{2}$Control/Robotics Research Laboratory (CRRL) of Electrical and Computer Engineering Department, New York University, 5 Metrotech Center, Brooklyn, NY, USA.}
\thanks{$^(3)$ Electrical Engineering, New York University Abu Dhabi (NYUAD), Abu Dhabi, United Arab Emirates (UAE).}
\thanks{This paper has been accepted at the IEEE Conference on Decision and Control (CDC 2025), Rio de Janeiro, Brazil. © 2025 IEEE.}
}

\begin{document}

\maketitle
\thispagestyle{empty}
\pagestyle{empty}

\begin{abstract}
	In this paper, a consensus algorithm is proposed for interacting multi-agents, which can be modeled as simple Mechanical Control Systems (MCS) evolving on a general Lie group. The standard Laplacian flow consensus algorithm for double integrator systems evolving on Euclidean spaces is extended to a general Lie group. A tracking error function is defined on a general smooth manifold for measuring the error between the configurations of two interacting agents. The stability of the desired consensus equilibrium is proved using a generalized version of Lyapunov theory and LaSalle's invariance principle applicable for systems evolving on a smooth manifold. The proposed consensus control input requires only the configuration information of the neighboring agents and does not require their velocities and inertia tensors. The design of tracking error function and consensus control inputs are demonstrated through an application of attitude consensus problem for multiple communicating rigid bodies. The consensus algorithm is numerically validated by demonstrating the attitude consensus problem.
\end{abstract}

\section{Introduction}
Consensus algorithms are a category of distributed algorithms for facilitating a collective agreement among a set of agents, with varying initial conditions and incomplete state information of all the agents. These algorithms have been extensively researched for scenarios, where multiple agents operate within a finite-dimensional Euclidean space \cite{olfati2004consensus,ren2008distributed,qin2016recent,ren2005consensus,ren2005survey}. However, the configuration space of the majority of the mechanical, aerospace, and underwater systems cannot be represented by such a space, but rather as a Lie group. A survey of consensus algorithms on nonlinear spaces is documented in \cite{sepulchre2011consensus}. 

A large body of works discussing consensus algorithms for mechanical systems evolving on special examples of matrix Lie groups can be found such as $\mathrm{SE}\left(3\right)$ \cite{sun2017formation,nazari2016decentralized,igarashi2009passivity,dong2013consensus}, $\mathrm{SE}\left(2\right)$ \cite{sarlette2010coordinated,dong2013consensus,mirzaei2017robust} and $\mathrm{SO}\left(3\right)$ \cite{sahoo2014attitude,weng2013coordinated,CHAVAN2020109262,qi2023distributed,zhu2018adaptive}. Recent works on $\mathrm{SO}(3)$ include \cite{tang2022event,deng2021attitude} while for $\mathrm{SE}\left(3\right)$ \cite{lee2022position,thunberg2016consensus} can be traced.  However, these works are not applicable to a general Lie group. There are attempts to generalize these consensus algorithms on $\mathrm{SE}(n), \mathrm{SO}(n)$ to a general Lie group.  A unified geometric framework for coordinated motion on Lie groups is presented in the works of Sarlette et al. \cite{sarlette2008coordination,sarlette2010coordinated} using simple integrator dynamical systems, and cannot be directly applied to mechanical systems. To address this issue, a consensus algorithm for a Lie group with a bi-invariant metric is studied in \cite{seshan2022geometric} assuming double integrator dynamics. However, the stability of the desired consensus equilibrium was only guaranteed for a communication graph having the structure of a line graph. This was addressed in the work of Seshan et al ~\cite{2022arXiv220900345S}, which considers a more general connected graph using a double integrator system emulating a simple mechanical system \cite{FB-ADL:04}. Seshan et al, \cite{2022arXiv220900345S} utilizes a G-polar Morse function for designing the controller, however the existence of such a function on Lie group is still an open question. In addition, Seshan et al, \cite{2022arXiv220900345S} assumes the graph to be non-directed and the general assumption of a bi-invariant Riemannian metric is very strong and restrictive leaving the results in \cite{2022arXiv220900345S} to a very narrow class of systems.

From these studies, it can be concluded that the literature lacks a comprehensive study about consensus algorithms on a general compact Lie group with both non-directed and directed communication graphs. We focus on developing a generalized consensus algorithm for a general connected compact Lie group using a communication structure having the structure of a connected tree graph. The contributions of the present work are:
\begin{enumerate}
	\item The standard second-order Laplacian flow consensus algorithm for Euclidean space is extended for MCS evolving on a general connected Lie group ensuring position and velocity consensus.
	\item The asymptotic stability of the consensus equilibrium for both undirected and directed communication graphs is proven using a generalized version of Lyapunov theory for smooth manifolds.
	\item A consensus controller is proposed so that all agents follow a predefined reference tracking, provided only a subset of agents having the information about the reference trajectory.
	\item The consensus algorithm is derived for the case of attitude consensus of multiple rigid bodies. The results are validated through numerical simulations.
\end{enumerate}
The rest of the paper is organized as follows. 
Section \ref{Sec:Consensus_Lie_groups} presents the key results on consensus algorithm on Lie groups, along with the necessary mathematical preliminaries. Section~\ref{Sec:Attitude_Consensus} presents the dynamics behind the attitude consensus problem of multiple rigid bodies, while Section \ref{Sec:Numerical_simulation} discusses the numerical simulation of that problem. The concluding remarks are discussed in the last Section.
\subsection{Notations}
\noindent The paper follows the notations used in \cite{FB-ADL:04}. The set of real numbers is represented by $\mathbb{R}$. The set of non-negative real numbers and the set of strictly positive real numbers are denoted by $\mathbb{R}_{\ge 0}$ and $\mathbb{R}_{> 0}$, respectively. The set of real $n$ vectors is in $\mathbb{R}^n$ and the set of $m \times n$ real matrices in $\mathbb{R}^{m \times n}$. The identity matrix of order $n$ is denoted by $\mathbf{I}_n$. The operations $\operatorname{tr}\left(\cdot\right)$ and $\operatorname{skew}\left(\cdot\right)$ represent the trace and skew symmetric part of a matrix, respectively. 

\section{Consensus Algorithms on Lie Groups}\label{Sec:Consensus_Lie_groups}

\noindent Consider a smooth manifold $Q$ and let a point $q \in Q$, then $\mathrm{T}_qQ$ represents the tangent space at $q$. Let $X \in \Gamma^{\infty}\left(\mathrm{T}Q\right)$ be a smooth vector field, then the integral curve for $X$ is a curve $\gamma : I \to Q$ at $q_0 \in Q$ satisfying $\gamma'\left(t\right)=X\left(\gamma \left(t\right)\right),~~ t \in I$. Now the flow of $X$ is defined as $\Phi^{X}_t\left(q_0\right)=\gamma\left(t\right)$.  Let for $X \in \Gamma^{\infty}\left(\mathrm{T}Q\right)$ and $f \in C^{\infty}\left(Q\right)$ be a smooth function on Q, then Lie derivative with respect to $X$ and denote by $\mathcal{L}_X : C^{\infty}\left(Q\right)\to \mathbb{R}$ defined as $\mathcal{L}_Xf\left(q\right)=\left\langle \operatorname{d}f\left(q\right);X\left(q\right) \right\rangle$, where $\operatorname{d}f\left(q\right)\in \mathrm{T}^{*}_qQ $ and $\mathrm{T}^{*}_qQ$ represents the dual of $\mathrm{T}_qQ$. Let $\left(G,\star\right)$ be a Lie group where $\star$ represents the binary operation. Let $\mathfrak{g}$ be its Lie algebra and $e$ be the identity element of $G$. The adjoint operator $\operatorname{ad}_\xi : \mathfrak{g} \to \mathfrak{g}$ is defined as $\operatorname{ad}_{\xi}\eta=\left[\xi,\eta\right],~~\forall \xi,\eta \in \mathfrak{g}$, where $\left[ \cdot,\cdot \right]$ is the Lie bracket operation on $\mathfrak{g}$. Next, defining the left translation map $\mathscr{L}_g: G\to G,~~g\in G$ as $\mathscr{L}_g\left(h\right)=g \star h,~~h \in G$. Similarly the right translation map $\mathscr{R}_g: G\to G,~~g\in G $ is defined as $\mathscr{R}_g\left(h\right)=h \star g$. Also given a left-invariant affine connection $\nabla$, there exists a unique bilinear map $B : \mathfrak{g} \times \mathfrak{g} \to \mathfrak{g}$ such that for all $\xi,\eta \in \mathfrak{g}$, we have $\nabla_{\xi_L}\eta_L= \left(B\left(\xi,\eta\right)\right)_L$. The converse of the above statement is also true.   Let $\Psi \in C^\infty \left(Q\right)$ and consider a point $q_0 \in Q$, then $q_0$ is a critical zero for $\Psi$, if $\Psi\left(q_0\right)=0$ and $\operatorname{d}\Psi\left(q_0\right)=0_{q_0}$. If all of the critical points of $\Psi$ are nondegenerate, then $\Psi$ is a Morse function. If $q_0 \in Q$ is a critical zero for a smooth function $\Psi$ and if the Hessian of $\Psi$ is positive-definite at $q_0$, then $\Psi$ is locally positive-definite about $q_0$. Conversely, if $\Psi$ is locally positive-definite about $q_0$, then $q_0$ is a critical zero for $\Psi$ and the Hessian of $\Psi$ is positive-definite at $q_0$. The function $\Psi$ is polar if it has a unique minimum. Let $L \in \mathbb{R}$, $\Psi \in C^{\infty}\left(Q\right)$ and $q_0 \in Q$, then the $L-$sublevel set of $\Psi$ is denoted by $\Psi^{-1}\left(\le L\right)$ and defined as $\Psi^{-1}\left(\left( -\infty,L\right)\right)=\left\lbrace q \in Q~ \vert~ \Psi\left(q\right) \le L \right\rbrace$. In order,  to express the error between two agents, we need to define a configuration error map which takes two configurations as  arguments and returns a scalar value measuring the error between the two configurations. Next, defining a function $C^{\infty}\left(Q\times Q\right) \ni \Psi : Q \times Q \to \mathbb{R}$. The function $\Psi$ is symmetric, if $\Psi\left(q,r\right)=\Psi\left(r,q\right)$ for all $\left(q,r\right) \in Q \times Q$. For a symmetric function $\Psi \in C^{\infty} \left(Q \times Q\right)$, we define $\Psi_r:Q\to \mathbb{R}$ as $\Psi_r\left(q\right)=\Psi\left(q,r\right)$. For such a $\Psi$, we define the differential with respect to first argument $\operatorname{d}_1\Psi$ and second argument $\operatorname{d}_2\Psi$ \cite{FB-ADL:04}. Define a quantity $\tau_L \in \mathbb{R}$ as $\tau_L=	\operatorname{sup}\left\lbrace L \in \mathbb{R} ~\vert~ g\in \Psi^{-1}_G\left(\le L\right)\backslash \left\lbrace e\right\rbrace,\operatorname{d}\Psi_G \left(g\right) \neq 0\right\rbrace$, where $\Psi^{-1}_G\left(\le L\right)$ is a sublevel set containing $e$. Consider $n$ mutually interacting identical agents represented as MCSs, each evolving on a Lie group $G$. Then, the configuration space of the whole system is denoted by $ \begin{array}{cc}
	G^{n} \coloneqq & \underbrace{G \times  \dots \times G} \\
	& n- ~\text{times}
\end{array}$ and the set of body velocities for the whole system is denoted by $ \begin{array}{cc}
	\mathfrak{g}^{n} \coloneqq & \underbrace{\mathfrak{g} \times  \dots \times \mathfrak{g}} \\
	& n- ~\text{times}
\end{array}$. The communication network between $n$ multi-agents can be represented by a graph $\mathcal{G}=\left(\mathcal{V},\mathcal{E}\right)$ where the agents are denoted by the set of nodes $\mathcal{V}=\left\lbrace 1,2,\dots ,n\right\rbrace$ and $\mathcal{E} \subseteq \mathcal{V} \times \mathcal{V}$ denotes the set of directed edges. A directed edge $e_{ji} \in \mathcal{E}$ denotes that the agent $i$ is receiving information from agent $j$. Then, agent $j$ is labeled as a neighbor of agent $i$. The set of neighbors of agent $i$ is denoted by $\mathcal{N}_i$. If $\mathcal{G}$ is undirected then $e_{ji} \Leftrightarrow e_{ij} \in \mathcal{E},~\forall i,j \in \mathcal{V}$. The interaction between multi-agents can be expressed by the adjacency matrix $A=\left[a_{ij}\right]_{n \times n} \in \mathbb{R}^{n \times n}$, where $a_{ij}$ is defined as
$
a_{ij}=\begin{cases}
	1 & ,~\text{if}~j\in \mathcal{N}_i\\
	0 &,~\text{otherwise}
\end{cases}
$.
The degree of an agent $i$ is the number of edges that connect the agent $i$. If $\mathcal{G}$ is undirected, then $A$ is a symmetric matrix and the degree of agent $i$ is $\delta_i\triangleq \sum^{n}_{j=1}a_{ij}$. In a digraph $\mathcal{G}_{\operatorname{dir}}$ with an edge $e_{ij} \in \mathcal{E}_{\operatorname{dir}}$, then agent $i$ is called an in-neighbor of agent $j$ and $j$ is called the out-neighbor of agent $i$. Let $\mathcal{N}^{\operatorname{in}}_i$ be the set of in-neighbors of $i$-th agent. The adjacency matrix $A_{\operatorname{dir}}=\left[\tilde{a}_{ij}\right]_{n \times n} \in \mathbb{R}^{n \times n}$, where $a_{ij}$ is defined as
$
\tilde{a}_{ij}=\begin{cases}
	1 & ,~\text{if}~j\in \mathcal{N}^{\operatorname{in}}_i\\
	0 &,~\text{otherwise}
\end{cases}
$. Let the in-degree of $i-$ th agent is $\delta^{\operatorname{in}}_i$ and the out-degree of $i-$ th agent be $\delta^{\operatorname{out}}_i$. In the case of a directed graph the adjacency matrix need not be symmetric
.  Let $\gamma_i:\mathbb{R}_{\ge 0}\to G$ be the controlled trajectory of the $i$th agent and its body velocity is represented by the curve $t \mapsto v_i\left(t\right)=\mathrm{T}_{\gamma_i \left(t\right)}\mathscr{L}_{\gamma^{-1}_i\left(t\right)}\left(\gamma'_i\left(t\right)\right) \in \mathfrak{g}$. Then, the equations of motion of the $i$ th agent is given by \cite{FB-ADL:04}:
\begin{equation}\label{Eq:Dynamics_SMS_on_Lie_Group}
	\begin{split}
		\gamma'_i&=\mathrm{T}_e\mathscr{L}_{\gamma_i}\left(v_i\right)\\
		v'_i-\mathbb{I}^{\sharp}\left(\operatorname{ad}^{*}_{v_i}\mathbb{I}^{\flat}\left(v_i\right)\right)&= \mathbb{I}^{\sharp}\left(u_i\right)
	\end{split}
\end{equation}
where the control input $u_i\left(t,\gamma_i\left(t\right),\gamma_j\left(t\right),v_i\left(t\right)\right) \in \mathfrak{g}^{*},~ j \in \mathcal{N}_i$ and $\mathbb{I}$ represents the kinetic energy tensor. The flat map $\mathbb{I}^{\flat}$ and sharp map $\mathbb{I}^{\sharp}$ are defined as $\mathbb{I}^{\flat}:\mathfrak{g} \to \mathfrak{g}^{*}$ and $\mathbb{I}^{\sharp}:\mathfrak{g}^{*} \to \mathfrak{g}$ Let the position and velocity of all agents is represented by $\xi= \left(g_1,\dots,g_n\right)\times \left(v_1,\dots,v_n\right)\in G^{n} \times \mathfrak{g}^{n}$. Next, we consider the problem of achieving consensus in position of agents. The definition of the operator $\left(\cdot \right)^{\sharp}$ is provided earlier
. The problem objective is:
\vspace*{-2mm}\subsection*{Consensus objective I}
%
\noindent The consensus objective is to ensure
\begin{equation}\label{Eq:Consensus_Obj_1_Lie_Group}
\begin{split}
	\lim_{t \mapsto \infty} \gamma_i\left(t\right)=\gamma_j\left(t\right),\lim_{t \mapsto \infty}v_i\left(t\right)=0,~\forall i,j\in \mathcal{V}.\\
\end{split}
\end{equation}
 Consider a smooth function $\Psi_G: G\to \mathbb{R}_{\ge 0}$ which is proper, bounded from below satisfying: (1) $\Psi_G\left(e\right)=0$, 
(2) $\operatorname{d}\Psi_G\left(e\right)=0^{*}_e$, and
(3) $\operatorname{Hess}\Psi_G\left(e\right)$ is positive-definite. It can be concluded that $\Psi_G$ is a locally positive-definite function about $e$ and sublevel sets of $\Psi_G$ are compact. The set of critical points of $\Psi_G$, denoted by $\mathcal{S}^{\operatorname{crit}}_G$ is defined as $\mathcal{S}^{\operatorname{crit}}_G =\left\lbrace g \in G ~\vert~\operatorname{d}\Psi_G\left(g\right)=0^{*}_e\right\rbrace$.  Let the right tracking error function $\Psi_{G,r}: G \times G \to \mathbb{R}$ be
$\Psi_{G,r}\left(g_1,g_2\right)=\Psi_G\left(g^{-1}_2\star g_1\right)$. 
$\Psi_{G,r}$ is symmetric and we can obtain $\Psi_{G,r}\left(g_1,g_2\right)=\Psi_G\left(g^{-1}_2\star g_1\right)=\Psi_G\left(g^{-1}_1\star g_2\right)=\Psi_{G,r}\left(g_2,g_1\right)$. Let $\gamma_i\left(t\right), \gamma_j\left(t\right) \in G$ be the configurations of two agents, $i,j \in \mathcal{V}$, then the tracking error function between these two agents is given by $\Psi_{G,r}\left(\gamma_i\left(t\right),\gamma_j\left(t\right)\right)=\Psi_G\left(\gamma^{-1}_j\left(t\right)\star \gamma_i\left(t\right)\right)$. To compare the velocities of two agents, we use the adjoint map. Let $v_i$ and $v_j$ be the velocities of $i-$th and $j-$th agent defined in the Lie algebra $\mathfrak{g}$, then the error in velocity is defined as $v_i-\operatorname{Ad}_{\gamma^{-1}_i\gamma_j}\left(v_j\right)$. The velocity error in $\mathrm{T}G$ as $\dot{e}_{ij} \in \mathrm{T}_{\gamma_i} G$  is defined as $	\dot{e}_{ij}=\mathrm{T}_e\mathscr{L}_{\gamma_i}\left(v_i-\operatorname{Ad}_{\gamma^{-1}_i\gamma_j}\left(v_j\right)\right).$
Based on the definition of $\dot{e}_{ij}$, the time derivative of $\Psi_{G,r}\left(\gamma_i\left(t\right),\gamma_j\left(t\right)\right)$ is
\begin{equation}\label{Eq:Psi_{G,r}_dot_exp}
\begin{split}
	 \mathcal{L}_{X}\Psi_{G,r}\left(\gamma_i,\gamma_j\right)&=\left\langle \left(\mathrm{T}_e\mathscr{L}_{\gamma^{-1}_j\gamma_i}\right)^{*}\left(\operatorname{d}\Psi_{G}\left(\gamma^{-1}_j\gamma_i\right)\right);\right.\\
		&\qquad\left.\mathrm{T}_e\mathscr{L}_{\gamma^{-1}_i}\left(\dot{e}_{ij}\right) \right\rangle.
	\end{split}
\end{equation}	
The suggested control input $u_i\left(t\right)$ for the $i-$th agent is 
\begin{equation}\label{Eq:Control_input_Consensus_Lie_group}
\begin{split}
		\small u_i &=-\operatorname{ad}^{*}_{v_i}\mathbb{I}^{\flat}\left(v_i\right)-\frac{1}{\alpha_i} \left(
		R^{\flat}_{\operatorname{diss}}\left(v_i\right) \right.\\
		&\quad\left.+K_{p} \sum_{j=1}^{n}a_{ij}\left(\mathrm{T}_e \mathscr{L}_{\gamma_i}\right)^{*}\left(\operatorname{d}_i \Psi_{G,r}\left(\gamma_i,\gamma_j\right)\right)\right)
	\end{split}
\end{equation}
where  $K_p,\alpha_i \in \mathbb{R}_{>0}$ and $R_{\operatorname{diss}}:\mathfrak{g} \times \mathfrak{g} \to \mathbb{R}_{\ge 0}$ is a positive semi-definite bilinear mapping . Substituting the control input $u_i$ in \eqref{Eq:Dynamics_SMS_on_Lie_Group}, the closed-loop dynamics on $\mathrm{T}G^{n}$ is
\begin{equation}\label{Eq:Consensus_Closed_loop_Dynamics_Lie_group}
	\begin{split}
		\gamma'_i&=\mathrm{T}_e\mathscr{L}_{\gamma_i}\left(v_i\right)\\
		v'_i&=-\frac{\mathbb{I}^{\sharp}}{\alpha_i}\left(K_{p} \sum_{j=1}^{n}a_{ij}\left(\mathrm{T}_e \mathscr{L}_{\gamma_i}\right)^{*}\left(\operatorname{d}_i \Psi_{G,r}\left(\gamma_i,\gamma_j\right)\right)+R^{\flat}_{\operatorname{diss}}\left(v_i\right)\right).
	\end{split}
\end{equation}
The equilibrium set of the closed-loop dynamics is
\begin{equation}\label{Eq:Eq_set_Lie_group}
	\small \begin{split}
		&\mathcal{E}^{1}_G=\left\lbrace\xi \in G^{n} \times \mathfrak{g}^{n} \vert
	\sum_{j=1}^{n}a_{ij}\operatorname{d}_i\Psi_{G,r}\left(g_i,g_j\right)=0^{*}_{g_i},v_i=0, \forall i\in \mathcal{V} \right\rbrace.
	\end{split}
\end{equation}
The desired consensus equilibrium is 
\begin{equation}
\begin{split}
	\mathcal{E}^{*,1}_G&=\left\lbrace \xi \in G^{n} \times \mathfrak{g}^{n} ~\vert~
g_1=\dots=g_n,v_1=\dots=v_n=0\right\rbrace.
\end{split}
\end{equation}
Similarly, there exists undesired consensus equilibria defined as $\mathcal{\bar{E}}^{1}_G=\mathcal{E}^{1}_G  \backslash \mathcal{E}^{*,1}_G$. The undesired consensus equilibria consists of trivial $\mathcal{\bar{E}}^{\operatorname{triv},1}_G$ and nontrivial equilibria $\mathcal{\bar{E}}^{\operatorname{nontriv},1}_G$ are
\begin{equation}
\begin{split}
	&\mathcal{\bar{E}}^{\operatorname{triv},1}_G=\left\lbrace \xi \in \mathcal{E}^{\operatorname{des},1}_G \vert
 g^{-1}_j  \star g_i \in \mathcal{S}^{\operatorname{crit}}_G,v_i=0,\forall i,j \in \mathcal{V}\right\rbrace.
\end{split}
\end{equation}
and $\mathcal{\bar{E}}^{\operatorname{nontriv},1}_G=\mathcal{\bar{E}}^{1}_G \backslash \mathcal{\bar{E}}^{\operatorname{triv},1}_G $. The next theorem establishes the asymptotic stability of the consensus equilibrium $	\mathcal{E}^{*,1}_G$ under the control input $u_i$ given in \eqref{Eq:Control_input_Consensus_Lie_group}.
\vspace*{-2mm}
\begin{theorem}\label{Th:Consensus_Lie_group_Obj_1}
	Let the communication graph $\mathcal{G}$ be connected tree and the initial conditions lie in the set 
\begin{equation}
\small 	\begin{split}
	&\mathcal{S}^{1}=\left\lbrace\xi \in G^n \times \mathfrak{g}^n ~\vert~ 
	\frac{K_{p}}{2} \sum_{i=1}^{n}\sum_{j=1}^{n}a_{ij}\Psi_{G,r}\left(g_i\left(0\right),g_j\left(0\right)\right)\right.\\
&\left.+\frac{1}{2}\sum_{i=1}^{n}\alpha_i \left\langle \left\langle v_i\left(0\right),v_i \left(0\right)\right\rangle \right\rangle_{\mathbb{I}} \le \frac{K_p}{2}\tau_L\right\rbrace.
\end{split}
\end{equation}
Then the proposed control input $u_i\left(t\right)$ given in \eqref{Eq:Control_input_Consensus_Lie_group} renders the consensus equilibrium $\mathcal{E}^{\operatorname{des},1}_G$ locally asymptotically stable.
\end{theorem}
\vspace*{-7mm}
\begin{proof}
	See Appendix \ref{App:Proof_Th1}.
\end{proof}
\noindent Next, we consider the consensus in position and velocity of each agent. The problem objective is stated as:
\subsection*{Consensus objective II}
\vspace*{-3mm}
\begin{equation}\label{Eq:Consensus_Obj_2_Lie_Group}
\begin{split}
	\lim_{t \mapsto \infty} \gamma_i\left(t\right)=\gamma_j\left(t\right),	\lim_{t \mapsto \infty}v_i\left(t\right)=v_j\left(t\right),~\forall i,j\in \mathcal{V}.\\
\end{split}
\end{equation}
\hspace*{2mm}The total body velocity error for the $i$th agent is
\begin{equation}
\mathfrak{g} \ni v^{e}_i =\sum^{n}_{j=1}a_{ij} \left(v_i -\operatorname{Ad}_{\gamma^{-1}_i\gamma_j}\left(v_j\right)\right).
\end{equation}
A new error vector $\sigma_i\in \mathfrak{g}$ is defined, which is a linear combination of the velocity error $v^{e}_i $ and position error $\sum_{j=1}^{n}a_{ij}\left(\mathrm{T}_e \mathscr{L}_{\gamma_i}\right)^{*}\left(\operatorname{d}_i \Psi_{G,r}\left(\gamma_i,\gamma_j\right)\right)$ as
\begin{equation}
	\sigma_i =v^{e}_i + K_p \mathbb{I}^{\sharp} \Bigg(\sum_{j=1}^{n}a_{ij}\left(\mathrm{T}_e \mathscr{L}_{\gamma_i}\right)^{*}\left(\operatorname{d}_i \Psi_{G,r}\left(\gamma_i,\gamma_j\right)\right)\Bigg).
\end{equation}
The control input is given by
\begin{equation}\label{Eq:Control_input_Lie_group_Obj_2}
	\small 	\begin{split}
		&u_i = \frac{1}{\delta_i} \left( -\frac{1}{\alpha_i} \left( K_{p}\sum_{j=1}^{n}a_{ij} \left(\mathrm{T}_e\mathscr{L}_{\gamma^{-1}_j\gamma_i}\right)^{*}\left(\operatorname{d}\Psi_G\left(\gamma^{-1}_j\gamma_i\right)\right)\right.\right.\\
		&\left.\left. R^{\flat}_{\operatorname{diss}}\left(\sigma_i\right)\right)+\sum_{j=1}^{n}a_{ij} \mathbb{I}^{\flat}\left(\stackrel{\mathfrak{g}}{\nabla}_v \left(\operatorname{Ad}_{\gamma^{-1}_i\gamma_j}\left(v_j\right)\right)\right. \right.\\
		&\left. \left.+\left[\operatorname{Ad}_{\gamma^{-1}_i\gamma_j}\left(v_j\right),v_i\right] +\operatorname{Ad}_{\gamma^{-1}_i\gamma_j}\left(v'_j\right)\right)-\operatorname{ad}^{*}_{v_i}\mathbb{I}^{\flat}\left(v_i\right)\right.\\
		&\left. -K_p\frac{\operatorname{d}}{\operatorname{d}t} \left(\mathbb{I}^{\sharp} \sum_{j=1}^{n}a_{ij}\left(\mathrm{T}_e \mathscr{L}_{\gamma_i}\right)^{*}\left(\operatorname{d}_i \Psi_{G,r}\left(\gamma_i,\gamma_j\right)\right)\right) \right).
	\end{split}
\end{equation}
Substituting the control input in the dynamics of $\sigma_i$, we obtain
\begin{equation}\
\small
\hspace*{-2mm}
\begin{split}
	\sigma'_i=-\frac{\mathbb{I}^{\sharp}}{\alpha_i}\left(K_{p} \sum_{j=1}^{n}a_{ij}\left(\mathrm{T}_e \mathscr{L}_{\gamma_i}\right)^{*}\left(\operatorname{d}_i \Psi_{G,r}\left(\gamma_i,\gamma_j\right)\right)+R^{\flat}_{\operatorname{diss}}\left(\sigma_i\right)\right).
\end{split}
\end{equation}
The desired consensus equilibrium of the closed-loop dynamics is 
\begin{equation}
\begin{split}
	\mathcal{E}^{*,2}_G&=\left\lbrace \xi \in G^{n} \times \mathfrak{g}^n \vert g_1=\dots=g_n,  v_1=\dots=v_n\right\rbrace.
\end{split}
\end{equation}

The next theorem establishes the asymptotic stability of 	$\mathcal{E}^{*,2}_G$ under the control input $u_i\left(t\right)$ given in \eqref{Eq:Control_input_Lie_group_Obj_2}.
\vspace*{-3mm}
\begin{theorem}\label{Th:Consensus_Lie_group_Obj_2}
	Let the communication graph $\mathcal{G}$ be connected tree and the initial conditions lie in the set 
\begin{equation}
		\begin{split}
		&\mathcal{S}^{2}=\left\lbrace \xi \in G^{n} \times \mathfrak{g}^{n} ~\vert~
	 K_pa_{ij}\Psi_{G,r}\left(g_i\left(0\right),g_j\left(0\right)\right)\right.\\
		&\left.+\frac{1}{2}\alpha_i\left\langle \left\langle \sigma_i\left(0\right),\sigma_i \left(0\right)\right\rangle \right\rangle_{\mathbb{I}} \le \tau_L,\forall i \in \mathcal{V},j \in \mathcal{N}_i\right\rbrace.
	\end{split}
\end{equation}
Then the proposed control input $u_i\left(t\right)$ given in \eqref{Eq:Control_input_Lie_group_Obj_2} renders the consensus equilibrium $	\mathcal{E}^{*,2}_G$ locally asymptotically stable.
\end{theorem}
\vspace*{-7mm}
\begin{proof}
	See Appendix \ref{App:Proof_Th2}.
\end{proof}
\vspace*{-1mm}
\noindent The derived controllers in the previous theorems, there was no explicit way of specifying the final consensus value or trajectory. In the next section, we address this problem by proposing a consensus controller for reference trajectory tracking in which the agents converge to a specified value provided the information about reference value is available only to a subset of agents.
\vspace*{-1mm}
\subsection*{Consensus objective III}
\vspace*{-4mm}
\begin{equation}
	\begin{split}
		\lim_{t \mapsto \infty}  \gamma_i \left(t\right)=\lim_{t \mapsto \infty}  \gamma_r \left(t\right),~\lim_{t \mapsto \infty} v_i\left(t\right)=v_r \left(t\right).
	\end{split}
\end{equation}
where $\gamma_r()t$ is the twice differentiable reference trajectory with bounded velocity. Let the communication graph be directed; defining the total body velocity error for the $i-$th agent as
\begin{equation}
\mathfrak{g} \ni v^{e}_i =\sum^{n+1}_{j=1}\tilde{a}_{ij} \left(v_i -\operatorname{Ad}_{\gamma^{-1}_i\gamma_j}\left(v_j\right)\right).
\end{equation}
The new error vector $\sigma_i\in \mathfrak{g}$ is defined, as
\begin{equation}
	\sigma_i =v^{e}_i + K_p \mathbb{I}^{\sharp} \sum_{j=1}^{n+1}a_{ij}\left(\mathrm{T}_e \mathscr{L}_{\gamma_i}\right)^{*}\left(\operatorname{d}_i \Psi_{G,r}\left(\gamma_i,\gamma_j\right)\right).
\end{equation}
Next, we propose a controller with proportional, derivative and integral term as
\begin{equation}\label{Eq:Control_input_Lie_group_Obj_3}
	\begin{split}
		&u_i = \frac{1}{\delta^{\operatorname{in}}_i} \Bigg(
		-\frac{1}{\alpha_i} \Bigg[
		K_{p}\sum_{j=1}^{n+1}\tilde{a}_{ij} 
		\Big( \mathrm{T}_e\mathscr{L}_{\gamma^{-1}_j\gamma_i} \Big)^{*}
		\big( \operatorname{d}\Psi_G(\gamma^{-1}_j\gamma_i) \big) 
		\\
		&+ R^{\flat}_{\operatorname{diss}}(\sigma_i)
		\Bigg]  + \sum_{j=1}^{n+1} \tilde{a}_{ij} \Bigg(
		\mathbb{I}^{\flat}\Big(
		\stackrel{\mathfrak{g}}{\nabla}_v 
		\big( \operatorname{Ad}_{\gamma^{-1}_i\gamma_j}(v_j) \big)
		\\
		&+ \big[ \operatorname{Ad}_{\gamma^{-1}_i\gamma_j}(v_j), v_i \big]
		+ \operatorname{Ad}_{\gamma^{-1}_i\gamma_j}(v'_j)
		\Big)
		\Bigg) - \operatorname{ad}^{*}_{v_i}\mathbb{I}^{\flat}(v_i) \\
		& - K_p \frac{\operatorname{d}}{\operatorname{d}t} \Bigg(
		\mathbb{I}^{\sharp} \sum_{j=1}^{n+1} a_{ij}
		\Big( \mathrm{T}_e \mathscr{L}_{\gamma_i} \Big)^{*}
		\big( \operatorname{d}_i \Psi_{G,r}(\gamma_i,\gamma_j) \big)
		\Bigg)
		\Bigg).
	\end{split}
\end{equation}
where $K_p\in \mathbb{R}_{>0}$.  The closed-loop dynamics is 
\begin{equation}\label{Eq:Red_Order_Dyn_Obj_4}
\small
\hspace*{-5mm}
\begin{split}
	\sigma'_i=-\frac{\mathbb{I}^{\sharp}}{\alpha_i}\left(K_{p} \sum_{j=1}^{n+1}a_{ij}\left(\mathrm{T}_e \mathscr{L}_{\gamma_i}\right)^{*}\left(\operatorname{d}_i \Psi_{G,r}\left(\gamma_i,\gamma_j\right)\right)+R^{\flat}_{\operatorname{diss}}\left(\sigma_i\right)\right).
\end{split}
\end{equation}
The desired consensus equilibrium is 
\begin{equation}
\begin{split}
		\mathcal{E}^{*,3}_G&=\left\lbrace \xi \in G^n \times \mathfrak{g}^n~\vert~g_1=\dots=g_r, 
 v_1=\dots=v_r\right\rbrace.
\end{split}
\end{equation}
The next theorem establishes the asymptotic stability of $	\mathcal{E}^{*,3}_G$ under the control input $u_i\left(t\right)$ given in \eqref{Eq:Control_input_Lie_group_Obj_3}.
\vspace*{-2mm}
\begin{theorem}\label{Th:Consensus_Lie_group_Obj_III}
		Let the directed communication graph $\mathcal{G}_{dir,~n+1}$ be a directed out-tree with reference trajectory as the root agent  all the remaining agents have a maximum in-degree of 1 and the initial conditions lie in the set and the initial conditions lie in the set 
\begin{equation}
\begin{split}
\mathcal{S}^{3}&=\left\lbrace \xi \in G^n \times \mathfrak{g}^n~\vert~ K_p \tilde{a}_{ij}\Psi_{G,r}\left(g_i\left(0\right),g_j\left(0\right)\right)\right.\\
	&\left.+\frac{1}{2}\alpha_i \left\langle \left\langle \sigma_i\left(0\right),\sigma_i \left(0\right)\right\rangle \right\rangle_{\mathbb{G}} \le \tau_L,\forall i \in \mathcal{V},j \in \mathcal{N}_i\right\rbrace.
\end{split}
\end{equation}
Then the proposed control input $u_i\left(t\right)$ given in \eqref{Eq:Control_input_Lie_group_Obj_3} renders the consensus equilibrium $\mathcal{E}^{\operatorname{des},3}_G$ locally asymptotically stable.
\vspace*{-1.5mm}
\end{theorem}
\begin{proof}
	See Appendix \ref{App:Proof_Th3}.
\end{proof}
\section{Attitude consensus of multiple bodies}\label{Sec:Attitude_Consensus}
\noindent  In this section, we provide an example of attitude consensus of multiple agents modeled as rigid bodies. 
\noindent The proposed consensus algorithm for Lie group is utilized for developing an attitude consensus controller for multiple rigid bodies, so that the orientation of all the rigid bodies is aligned to a common one. Consider an inertial frame $\left(\mathcal{O}_G,e_1,e_2,e_3\right)$ and body frame of $i$th agent $\left(\mathcal{O}^{i}_B,b^{i}_1,b^{i}_2,b^{i}_3\right)$. The attitude of a rigid body is represented by a rotation matrix, which is an element of the special orthogonal group $\mathrm{SO}\left(3\right) \triangleq \left\lbrace R \in \mathbb{R}^{3 \times 3}|~RR^T=\mathbf{I}_3,~\text{det}\left(R\right)=1 \right\rbrace$. The rotation matrix $R_i\left(t\right) \in \mathrm{SO}\left(3\right)$ represents the attitude of agent $i$ with respect to the  inertial frame $\mathcal{O}_G$. The body angular velocity of the agent $i$ is denoted by $\omega_i:\mathbb{R}\to \mathbb{R}^3$ given by $\mathfrak{so}\left(3\right) \ni \hat{\omega}_i\left(t\right) = R^\top_i\left(t\right)\dot{R}_i\left(t\right)$. The hat map $\hat{\cdot}:\mathbb{R}^3 \to \mathfrak{so}\left(3\right)$ defined as
$\hat{x}y=x \times y,~\forall x,y \in \mathbb{R}^3$, where `$\times$' is the cross product in $\mathbb{R}^3$. The dynamics of the agent $i$ is
\begin{equation}\label{Eq:Multi_agent_rigid_body_dyn}
	\begin{split}
		\dot{R}_i &= R_i\hat{\omega}_i\\
		\dot{\omega}_i+J^{-1}\left(\hat{\omega}_iJ\omega_i\right) &= J^{-1}u_i.
	\end{split}
\hspace*{-2mm}
\end{equation}

The term $J\in \mathbb{R}^{3 \times 3}$ represents the inertia matrix of the agent $i$ and the control input $u_i \in \mathbb{R}^3 \simeq \mathfrak{so}^{*}\left(3\right)$. Given the current attitude $R_i$ of agent $i$ and that of its neighbor $j$, $R_j$, we define the right attitude error as $R_{ij}\left(t\right)=R^\top_j\left(t\right)R_i\left(t\right)$. The right tracking error function is defined $\Psi_{\text{SO}\left(3\right),r}:\text{SO}\left(3\right)\times \text{SO}\left(3\right) \to \mathbb{R}$ so as $	\Psi_{\text{SO}\left(3\right),r}=\frac{1}{2}\operatorname{tr}\left(P\left(\mathbf{I}_3-R_{ij}\right)\right).$
, where $P=\operatorname{diag}\left\lbrace p_1,p_2,p_3\right\rbrace$ with $p_i \in \mathbb{R}_{>0}$. The tracking function $\Psi_{\text{SO}\left(3\right),r}$ has four critical points, which is the minimum number of critical points as allowed by the topology of the configuration space $\mathrm{SO}\left(3\right)$. The critical points of $\Psi_{\mathrm{SO}\left(3\right),r}$ are $R^{\operatorname{crit}}_{e,r}\in \mathcal{S}^{\operatorname{crit}}_{\mathrm{SO}\left(3\right)}=\left\lbrace \mathbf{I}_3,\operatorname{exp}\left(\pi \hat{e}_1\right),\operatorname{exp}\left(\pi \hat{e}_2\right),\operatorname{exp}\left(\pi \hat{e}_3\right)\right\rbrace$. The desired critical point is $\mathbf{I}_3$ and all other critical points are orientations obtained by a rotation of angle $\pi$ from the desired equilibrium point $R_j$ about $e_i,~\forall~ i \in \left\lbrace 1,2,3\right\rbrace$. The time derivative of $\Psi_{\text{SO}\left(3\right),r}$ is given by
\vspace*{-2mm}
\begin{equation}\label{Eq:Psi_SO_3_dot_exp}
	\begin{split}
		\frac{\operatorname{d}}{\operatorname{dt}}\Psi_{\mathrm{SO}\left(3\right),r}\left(R_i,R_j\right) =\left(\operatorname{skew}\left(P R_{ij}\right)^\vee\right)^\top\omega_{ij}
	\end{split}
\end{equation}
where $\mathbb{R}^3 \ni \omega_{ij}\left(t\right)=\omega_i\left(t\right)-\left(R_{ij}\right)^\top\omega_j$ represents right velocity error in the body frame of agent $i$. Comparing \eqref{Eq:Psi_{G,r}_dot_exp} and \eqref{Eq:Psi_SO_3_dot_exp}, we have $\left(\mathrm{T}_{\mathbf{I}_3}\mathscr{L}_{R_{ij}}\right)^{*}\left(\operatorname{d}\Psi_{\mathrm{SO}\left(3\right)}\left(R_{ij}\right)\right)=\operatorname{skew}\left(PR_{ij}\left(t\right)\right)^\vee.$ The control input $u_i$ for agent $i$ in order to achieve consensus 
can be derived from \eqref{Eq:Control_input_Consensus_Lie_group}
\begin{equation}\label{Eq:Control_input_Consensus_SO_3}
	\begin{split}
		u_i&=-\reallywidehat{\omega_i}J\omega_i-\frac{1}{\alpha_i} \left(K_{p}\sum_{j=1}^{n}a_{ij}\left(\operatorname{skew}\left(PR_{ij}\right)^\vee\right) +K_D\omega_i\right)
	\end{split}
\end{equation}
where $\alpha_i,K_{p}, K_D \in \mathbb{R}_{>0}$. From Theorem \ref{Th:Consensus_Lie_group_Obj_2}, the control input $u_i\left(t\right)$ to achieve the objective given in \eqref{Eq:Consensus_Obj_2_Lie_Group} is
\begin{equation}
	\begin{split}
	&u_i= \frac{1}{\delta_i} \left(-\reallywidehat{\omega_i}J\omega_i-\frac{K_{p}}{\alpha_i}\sum_{j=1}^{n}a_{ij}\left(\operatorname{skew}\left(PR_{ij}\right)^\vee\right) \right. \\
	& \left.-\frac{K_D}{\alpha_i}\sum_{j=1}^{n}a_{ij} \left(\omega_{ij}+K_P J^{-1}\left(\operatorname{skew}\left(PR_{ij}\right)^\vee\right)\right)\right.\\
    &\left.+\sum_{j=1}^{n}a_{ij}  J^{-1}\left(\frac{1}{2}\left(\omega \times \left(R^\top_{ij}\omega_j\right)\right) \right.\right.\\
	&\left.\left.+\frac{1}{2}J^{-1}\left(\omega_i\times \left(J\left(R^\top_{ij}\omega_j\right)\right)+\left(R^\top_{ij}\omega_j\right)\times \left(J\omega_i\right)\right)\right)^\vee \right.\\
	&\left.+J^{-1}\left(\reallywidehat{R^\top_{ij}\omega_j}\widehat{\omega_i}-\widehat{\omega_i}\reallywidehat{R^\top_{ij}\omega_j}+\reallywidehat{R^\top_{ij}\dot{\omega}_j}\right)^\vee \right.\\
	&\left.-\frac{K_P}{2} J^{-1}\sum_{j=1}^{n}a_{ij}\left(\operatorname{tr} \left(R^\top_{ij}P\right)\mathbf{I}_3-R^\top_{ij} P\right)\omega_{ij}\right).
	\end{split}
\end{equation}
Consider the problem objective of tracking a reference attitude trajectory as
\begin{equation}\label{Eq:Consensus_Obj_SO_3_4}
	\begin{split}
		\lim_{t \mapsto \infty} R_i\left(t\right)=R_r\left(t\right),~\lim_{t \mapsto \infty}\omega_i\left(t\right)=\omega_r\left(t\right).~\forall i\in \mathcal{V}
	\end{split}
\end{equation}
where $R_r\left(t\right)$ and $\omega_r\left(t\right)$ represents the desired attitude trajectory and angular velocity. The control input $u_i$ to achieve this objective is given by
\begin{equation}
	\begin{split}
		\small &u_i=\frac{1}{\delta^{\operatorname{in}}_i} \left(-\frac{K_{p}}{\alpha_i}\left(\sum_{j=1}^{n+1}\tilde{a}_{ij}\left(\operatorname{skew}\left(PR_{ij}\right)^\vee\right)\right)\right. \\
		&\left.-\frac{K_D}{\alpha_i}\sum_{j=1}^{n+1}\tilde{a}_{ij} \left(\omega_{ij}+K_P J^{-1}_i\left(\operatorname{skew}\left(PR_{ij}\right)^\vee\right) \right)-\reallywidehat{\omega_i}J_i\omega_i\right.\\
		&\left.+\sum_{j=1}^{n+1}\tilde{a}_{ij}  J^{-1}\left(\frac{1}{2}\left(\omega \times \left(R^\top_{ij}\omega_j\right)\right) \right.\right.\\
		&\left.\left.+\frac{1}{2}J^{-1}\left(\omega_i\times \left(J\left(R^\top_{ij}\omega_j\right)\right)+\left(R^\top_{ij}\omega_j\right)\times \left(J\omega_i\right)\right)\right)^\vee \right.\\
		&\left.+J^{-1}\left(\reallywidehat{R^\top_{ij}\omega_j}\widehat{\omega_i}-\widehat{\omega_i}\reallywidehat{R^\top_{ij}\omega_j}+\reallywidehat{R^\top_{ij}\dot{\omega}_j}\right)^\vee \right.\\
		&\left.-\frac{K_P}{2} J^{-1}\sum_{j=1}^{n+1}\tilde{a}_{ij}\left(\operatorname{tr} \left(R^\top_{ij}P\right)\mathbf{I}_3-R^\top_{ij} P\right)\omega_{ij}\right).
	\end{split}
\end{equation}
\section{Simulation Studies}\label{Sec:Numerical_simulation}
\noindent 	A numerical simulation of the orientation and angular velocity consensus for multiple rigid bodies is presented inhere. The inertia matrix for all agents are diagonal and the parameters are identical for all agents with $J_{xx}=0.23 ,~J_{yy}=0.28$ and $J_{zz}=0.35$. Let $\left(\phi_i,\theta_i,\psi_i\right)$ represents the Euler angle of the $i$th agent, $\forall i=\left\lbrace 1,2,3,4 \right\rbrace$. 
\subsubsection{Consensus on orientation and angular velocity}
The problem objective is given in \eqref{Eq:Consensus_Obj_1_Lie_Group} and the communication graph is shown in Figure \ref{Fig:Graph}. The initial attitude and angular velocities for all agents are given in Table \ref{Table:IC}. The gains are chosen as $K_{p}=1,~K_D=2$. The simulation results are shown in Figures \ref{Fig:Euler_angle_Obj_2}-\ref{Fig:Control_input_SO_3_Obj_2}. The time histories of the Euler angles and angular velocities of the four agents are shown in Figure \ref{Fig:Euler_angle_Obj_2} and \ref{Fig:omega_Obj_2} and it can be observed that Euler angles and angular velocities of the four agents achieve consensus within 2 s. In this case it can be observed that the final consensus trajectories in Euler angles are time varying.  The control inputs for the four agents are shown in Figure \ref{Fig:Control_input_SO_3_Obj_2} and can be concluded that after the initial transient, the control input remains bounded.
\subsection{Consensus to a time varying reference attitude}
The problem objective is given in \eqref{Eq:Consensus_Obj_SO_3_4}, where the objective is track a time varying attitude trajectory. The reference angular velocity trajectory is chosen as $\omega_r \left(t\right)=\left[ 10 \operatorname{sin} \left(\frac{2\pi}{8}t\right),10 \operatorname{cos} \left(\frac{2\pi}{8}t\right),10 \operatorname{sin} \left(\frac{2\pi}{8}t\right)\right] \circ /s $. The reference attitude trajectory $R_r \left(t\right)$ is generated as $\dot{R}_r \left(t\right)=R_r\left(t\right) \hat{\omega_r} \left(t\right)$. The communication graph for this case is shown in Fig. \ref{Fig:Graph_tree} and the gains are chosen as $K_{p}=20,~K_D=10$ . The simulation results are shown in Figures \ref{Fig:Euler_angle_Obj_4}-\ref{Fig:Control_input_SO_3_Obj_4}. Figure \ref{Fig:Euler_angle_Obj_4} shows the Euler angles of all the agents and these converge to the reference trajectory. The time histories of the components of the angular velocities are shown in Figure \ref{Fig:omega_Obj_4} and converge to very small values within 3 seconds. The control inputs are shown in Fig.~\ref{Fig:Control_input_SO_3_Obj_4}.

\begin{table}[!htbp]
	\caption{Initial agent orientations and angular velocities}
	\label{Table:IC}
	\centering
	\begin{tabular}{ccc}
		\hline
		Agent no & $\left[\phi\left(0\right), \theta\left(0\right),\psi\left(0\right)\right]^\top$ &  $\omega\left(0\right)$\\
		\hline
		1 & $\left[20^{\circ}, 20^{\circ}  ,20^{\circ}\right]^\top$ & $\left[1,1,1\right]^\top\circ/s$\\
		
		2 & $\left[30^{\circ}, 30^{\circ} ,30^{\circ}\right]^\top$ & $\left[2,2,2\right]^\top\circ/s$\\
		
		3 & $\left[50^{\circ}, 50^{\circ}  ,50^{\circ}\right]^\top$ & $\left[2,2,2\right]^\top\circ/s$\\
		
		4 & $\left[70^{\circ}, 70^{\circ}  ,70^{\circ}\right]^\top$ & $\left[2,2,2\right]^\top\circ/s$\\
		\hline
	\end{tabular}
\end{table}
\begin{figure}[!htbp]
	\centering
	\begin{minipage}[b]{0.2\textwidth}
	\includegraphics[width=0.8\textwidth]{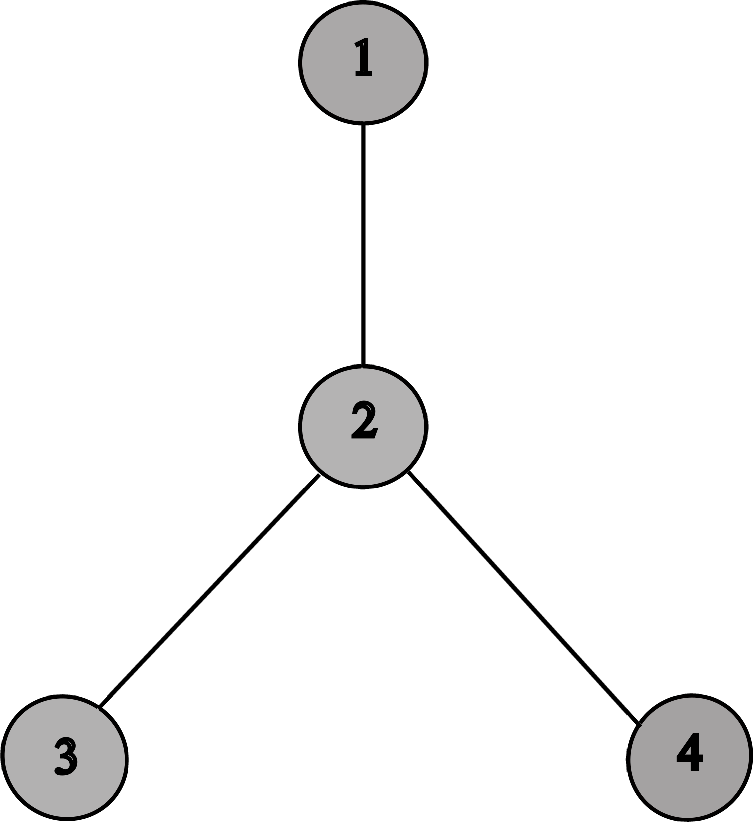}
	\caption{Communication graph $\mathcal{G}$}
	\label{Fig:Graph}
	\end{minipage}\hfill
	\begin{minipage}[b]{0.2\textwidth}
		\includegraphics[width=0.8\textwidth]{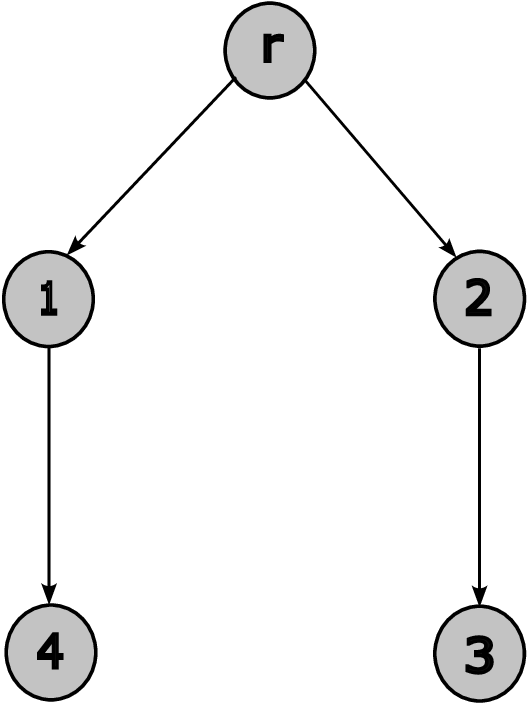}
		\caption{Connected tree communication graph }
		\label{Fig:Graph_tree}
	\end{minipage}
\end{figure}
\begin{figure}[!htbp]
	\centering
	\begin{minipage}[b]{0.5\textwidth}
		\includegraphics[width=0.85\textwidth]{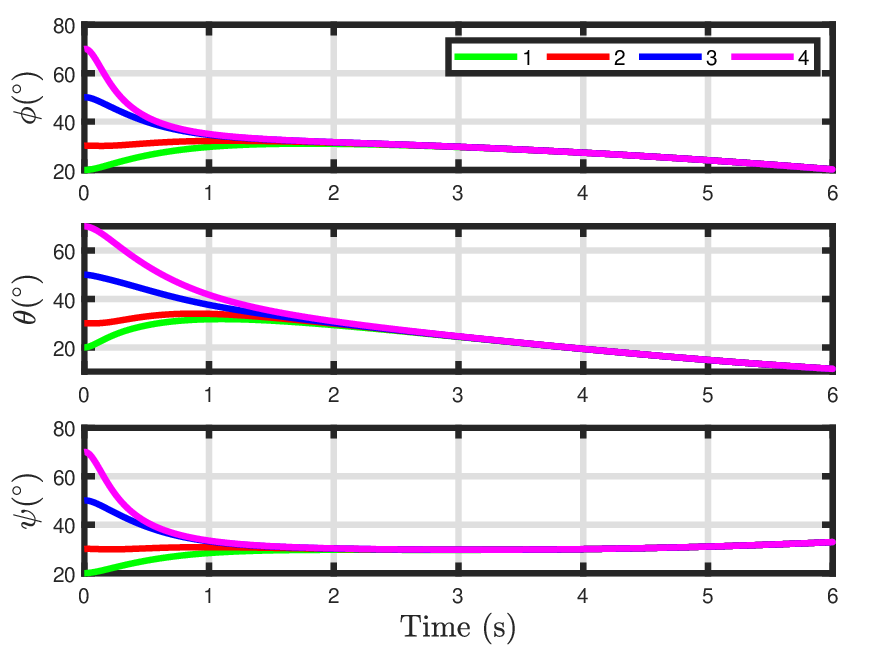}
		\caption{Euler angles responses}
		\label{Fig:Euler_angle_Obj_2}
	\end{minipage}
	\begin{minipage}[b]{0.5\textwidth}
		\includegraphics[width=0.85\textwidth]{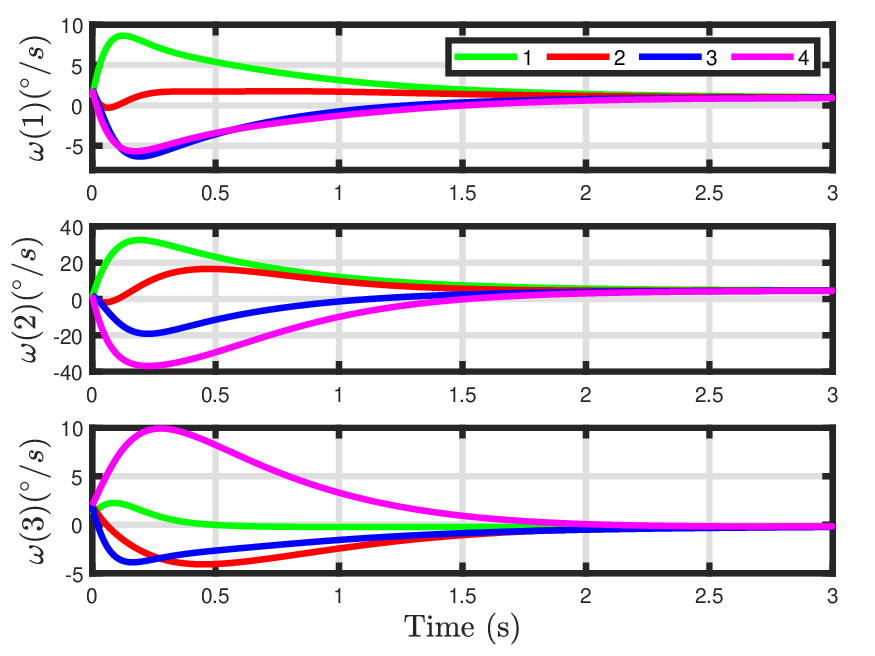}
		\caption{Angular velocities responses}
		\label{Fig:omega_Obj_2}
	\end{minipage}
    \begin{minipage}[b]{0.5\textwidth}
		\includegraphics[width=0.85\textwidth]{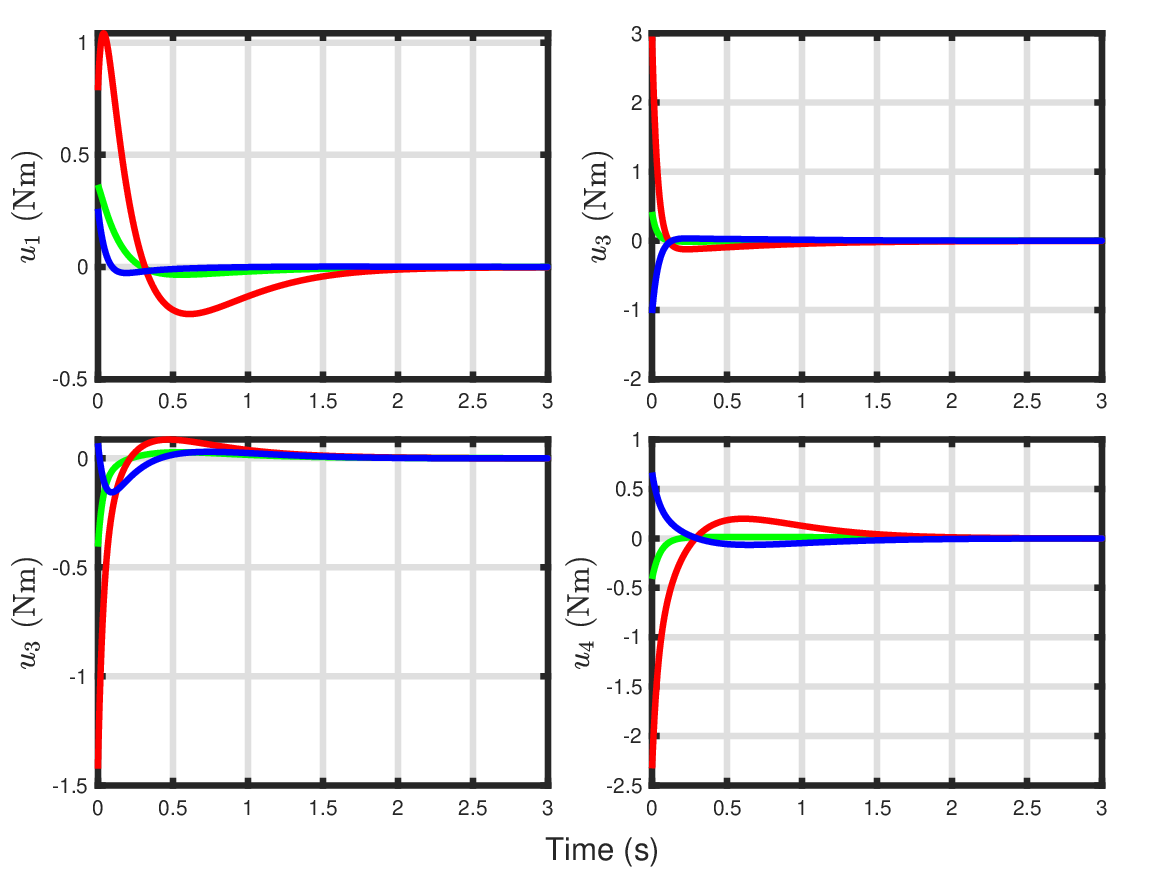}
		\caption{Control inputs responses}
		\label{Fig:Control_input_SO_3_Obj_2}
	\end{minipage}\hfill
		\begin{minipage}[b]{0.5\textwidth}
		\includegraphics[width=0.85\textwidth]{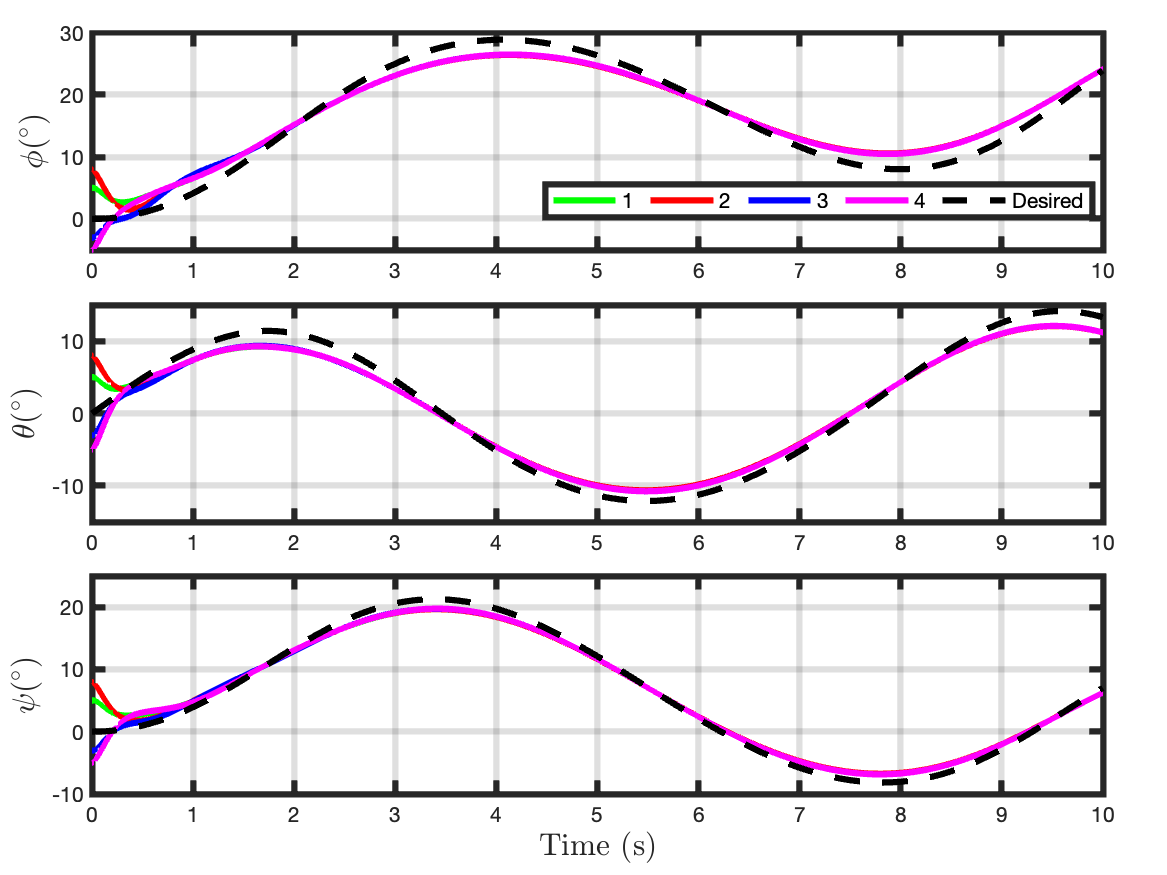}
		\caption{Euler angles responses}
		\label{Fig:Euler_angle_Obj_4}
	\end{minipage}\hfill
\end{figure}  
\begin{figure}[!htbp]
	\begin{minipage}[b]{0.5\textwidth}
	\includegraphics[width=0.95\textwidth]{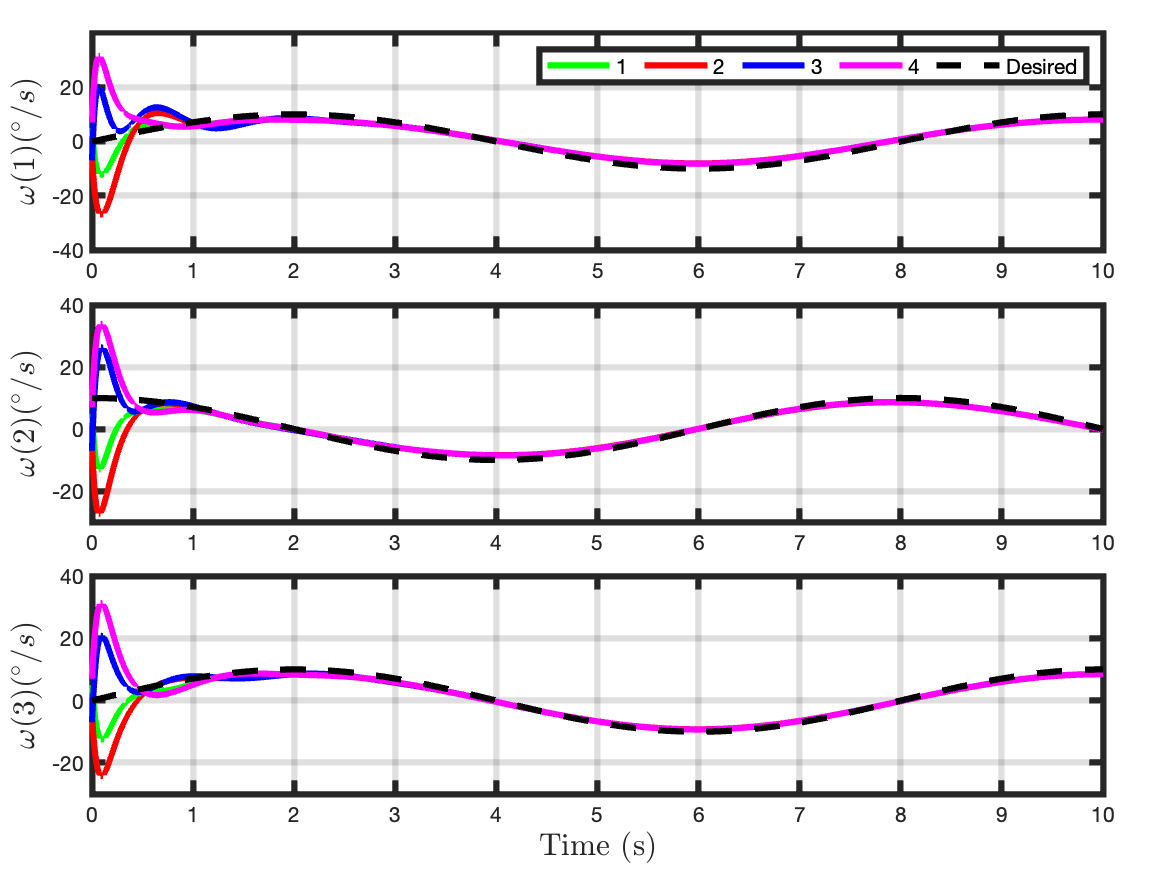}
	\caption{Angular velocities responses}
	\label{Fig:omega_Obj_4}
	\end{minipage}\hfill
	\begin{minipage}[b]{0.5\textwidth}
		\includegraphics[width=0.95\textwidth]{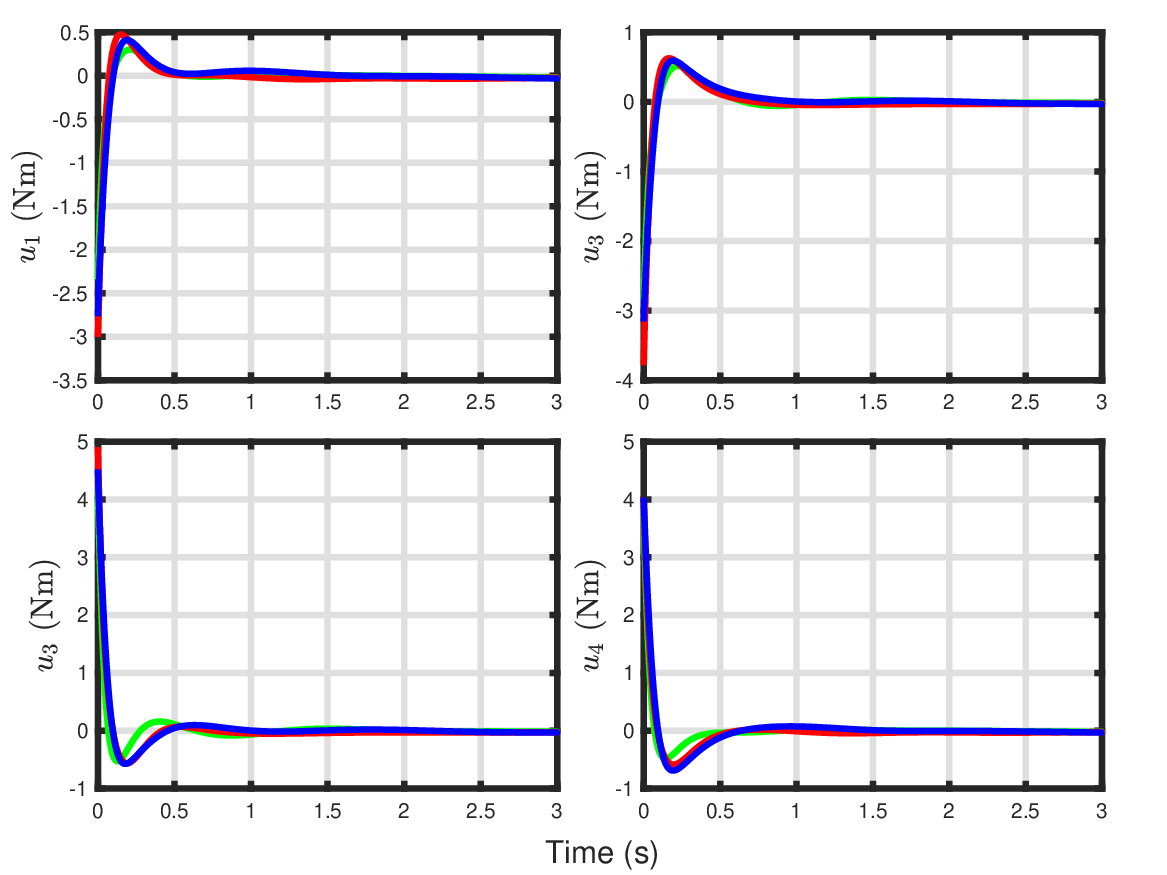}
		\caption{Control inputs responses}
		\label{Fig:Control_input_SO_3_Obj_4}
	\end{minipage}
\end{figure}  
\section*{Conclusions\label{Sec:Conclusion}}
\noindent 	In this paper, a general framework for designing consensus algorithm for multi-agent simple mechanical control systems evolving on a compact Lie group is proposed. A tracking error function is proposed for defining the error in position between two agents and adjoint map is used for comparing the velocities in the Lie algebra. The proposed controllers enable the communicating agents to achieve position as well as velocity consensus. In addition, controllers are proposed so that the agents track a predefined reference trajectory. Using Lyapunov theory for smooth manifold, local asymptotic stability of the consensus equilibrium is proved for both undirected and directed communication graphs.  The size of the region of convergence can be adjusted by changing the control parameters. The connectivity of the communication graph plays an important role in achieving consensus. The closed-loop dynamics resembles the structure of a second-order Laplacian flow for consensus in Euclidean space. Finally, the consensus algorithms are numerically validated by designing attitude consensus of rigid bodies. The future work plan is to extend the consensus algorithm to a general connected graph with switching networks.
 \section*{Acknowledgments}
 This work was supported by the NYUAD Center for Artificial Intelligence and Robotics (CAIR), funded by Tamkeen under the NYUAD Research Institute Award CG010.     
 \section*{Appendices}
 \subsection*{A. Proof of Theorem \ref{Th:Consensus_Lie_group_Obj_1}}{\label{App:Proof_Th1}}
 \begin{proof}
 	Consider a candidate Lyapunov function $V_1 \in C^{\infty}\left(\mathrm{T}G^{\operatorname{n}}\right)$ defined as
 	\begin{equation}
 		V_1=\frac{1}{2}K_{p}\sum_{i=1}^{n}\sum_{j=1}^{n}a_{ij}\Psi_{G,r}\left(g_i,g_j\right)+\frac{1}{2}\sum_{i=1}^{n}\alpha_i\left\langle \left\langle v_i,v_i \right\rangle \right\rangle_{\mathbb{I}},
 	\end{equation}
 	where $\left\langle\left\langle , \right\rangle\right\rangle_{\mathbb{I}}$ is the inner product on $\mathfrak{g}$ induced by the inertia $\mathbb{I}$. It can be easily verified that the candidate Lyapunov function $V_1$ is positive-definite about the consensus equilibrium $\mathcal{E}^{*,1}_G$. Let $X\in \Gamma^{\infty}\left(\mathrm{T}G^n\right)$ be the smooth vector field representing the closed-loop dynamics \eqref{Eq:Consensus_Closed_loop_Dynamics_Lie_group}. The Lie derivative of $V_1$ along \eqref{Eq:Consensus_Closed_loop_Dynamics_Lie_group}
 	\begin{equation}
 		\small	\begin{split}
 			&\mathcal{L}_XV_1
 			=\frac{1}{2}K_{p}\sum_{i=1}^{n}\sum_{j=1}^{n}a_{ij}\left(\left\langle \operatorname{d}_i \Psi_{G,r} \left(\gamma_i,\gamma_j\right); \gamma'_i \right \rangle \right.\\
 			&\left.+ \left\langle \operatorname{d}_j \Psi_{G,r} \left(\gamma_i \right),\gamma_j ;\gamma'_j\right\rangle \right)+\sum_{i=1}^{n}\alpha_i\left\langle \left\langle v_i,\mathbb{I}^{\sharp}_i \left(\operatorname{ad}^{*}_{v_i}\mathbb{I}^{\flat}_i\left(v_i\right)+u_i\right) \right \rangle \right\rangle	.
 		\end{split}
 	\end{equation} 
 	Since $\Psi_{G,r}$ is symmetric and we have $a_{ij}=a_{ji}$, the expression of $\mathcal{L}_XV_1$ can be simplified to
 	\begin{align}
 		&\mathcal{L}_X V_1=K_{p}\sum_{i=1}^{n}\sum_{j=1}^{n}a_{ij}\left(\left\langle \operatorname{d}_i \Psi_{G,r} \left(\gamma_i,\gamma_j\right); \gamma'_i \right \rangle \right)\nonumber \\
 		&+\sum_{i=1}^{n}\alpha_i\left\langle \left\langle v_i,\mathbb{I}^{\sharp}_i \left(\operatorname{ad}^{*}_{v_i}\mathbb{I}^{\flat}_i\left(v_i\right)+u_i\right) \right \rangle \right\rangle_{\mathbb{I}} \\
 		& =K_{p}\sum_{i=1}^{n}\sum_{j=1}^{n}a_{ij}\left(\left\langle \operatorname{d}_i \Psi_{G,r} \left(\gamma_i,\gamma_j\right); \mathrm{T}_e\mathscr{L}_{\gamma_i}\left(v_i\right) \right \rangle \right)\nonumber \\
 		&+\sum_{i=1}^{n}\alpha_i\left\langle \left\langle v_i,\mathbb{I}^{\sharp}_i \left(\operatorname{ad}^{*}_{v_i}\mathbb{I}^{\flat}_i\left(v_i\right)+u_i\right) \right \rangle \right\rangle_{\mathbb{I}}\\
 		&=K_{p}\sum_{i=1}^{n}\sum_{j=1}^{n}a_{ij}\left(\left\langle \left(\mathrm{T}_e \mathscr{L}_{\gamma_i }\right)^{*}\left(\operatorname{d}_i \Psi_{G,r}\left(\gamma_i,\gamma_j\right)\right);v_i \right \rangle \right)\nonumber \\
 		&+\sum_{i=1}^{n}\alpha_i\left\langle \left\langle v_i,\mathbb{I}^{\sharp}_i \left(\operatorname{ad}^{*}_{v_i}\mathbb{I}^{\flat}_i\left(v_i\right)+u_i\right) \right \rangle \right\rangle_{\mathbb{I}}.
 	\end{align}
 	Substituting the control input $u_i$ from \eqref{Eq:Control_input_Consensus_Lie_group} in the expression of $\mathcal{L}_X V_1$, we obtain $	\mathcal{L}_X V_1=-\sum_{i=1}^{n}R_{\operatorname{diss}}\left(v_i,v_i\right) \le 0$. By Lyapunov stability criteria for smooth manifolds \cite{FB-ADL:04} (Theorem 6.14), it can concluded that the consensus equilibrium $\mathcal{E}^{*,1}_G$ is stable. Also, let $V_1\left(t=0\right)=L_1 <\frac{K_P}{2}\tau_L$. Since $V_1\left(t\right)$ is a non-increasing function of time, the curve $t \mapsto \left(\gamma_i\left(t\right),\gamma_j\left(t\right)\right)$ lies in the set $\Psi^{-1}_{G,r}\left(\le L_1\right)$  for all $i,j \in \mathcal{V}$. Suppose the communication graph $\mathcal{G}$ is a connected tree. Let $l_i,~ i \in \left\lbrace 1,...,p\right\rbrace $ be a leaf of the tree  connected to the agent $\mathcal{N}_{l_{i}}$. Then by definition, the degree of $l_i$ is $1$ and $\mathcal{N}_{l_{i}}$ is the sole neighbor of $l_i$. Then, the non desired equilibrium points for agent $l_i$ is characterized by
 	\begin{equation}
 		\begin{split}
 			a_{l_i\mathcal{N}_{l_{i}}}\operatorname{d}_1\Psi_{G,r}\left(\gamma_{l_i}\left(t\right),\gamma_{\mathcal{N}_{l_i}}\left(t\right)\right)&=0^{*}_{\gamma_{l_i}\left(t\right)}
 			,~v_{l_i}\left(t\right)=0.
 		\end{split}
 	\end{equation}
 	Since $a_{l_i\mathcal{N}_{l_{i}}}\neq0$, we have for agent $l_i$ $\gamma^{-1}_{\mathcal{N}_{l_{i}}}\left(t\right)\star \gamma_{l_i}\left(t\right)\in \mathcal{S}^{\operatorname{crit}}_G$. Also, the initial conditions lie in the set $\mathcal{S}^1$ and the set $\mathcal{S}^1$ is positively invariant, it can be concluded that the position of both leaf and its sole neighbor $\mathcal{N}_{l_{i}}$, i.e, $\gamma_{l_i}\left(t\right) =\gamma_{\mathcal{N}_{l_{i}}}\left(t\right),~ \text{as}~ t \mapsto \infty$. Thus the positions of all leaves connected to $\mathcal{N}_{l_{i}}$ becomes identical. Consider another set of leaves $m_j,~ j \in \left\lbrace 1,...,q\right\rbrace $ connected to the agent $\mathcal{N}_{m_j}$. Following the above steps, it can be concluded that $\gamma_{m_j}\left(t\right) =\gamma_{\mathcal{N}_{m_{j}}}\left(t\right),~ \text{as}~ t \mapsto \infty$. Thus, the gradients associated with all the leaves vanishes. Consider a new subgraph $\tilde{\mathcal{G}}$ discarding the leaves of the graph $\mathcal{G}$. The agents $\mathcal{N}_{l_i} ,\mathcal{N}_{m_j},...$ become the leaves of the new subgraph $\tilde{\mathcal{G}}$. Following the above procedure, we end up in an empty graph in a finite number of steps and for any pair of agents $\gamma_i \left(t\right) =\gamma_j \left(t\right), \text{as}~ t \mapsto \infty,~\forall i,j \in \mathcal{V}$.
 \end{proof}
 \subsection{Proof of Theorem \ref{Th:Consensus_Lie_group_Obj_2}}{\label{App:Proof_Th2}}
 \begin{proof}
 	Let agent $k$ be a leaf connected to agent $l$. The closed-loop dynamics of $k-$ th agent is given by
 	\begin{equation}
 		\sigma'_k=-\frac{1}{\alpha_k} \mathbb{I}^{\sharp} \left(K_p \left(\mathrm{T}_e \mathscr{L}_{\gamma_k }\right)^{*} \left(\operatorname{d}_1 \Psi \left(\gamma_k,\gamma_l\right)\right)+R^{\flat}_{\operatorname{diss}}\left(\sigma_k\right) \right).
 	\end{equation}
 	The objective is to prove that $\gamma_k \left(t\right) \mapsto \gamma_l \left(t\right) $ and $\gamma'_k \left(t\right) \mapsto \gamma'_l \left(t\right)$ as $t \mapsto \infty$. Consider a candidate Lyapunov function $V_2: G \times \mathfrak{g} \to \mathbb{R}_{>0}$ as
 	\begin{equation}
 		V_2=K_p\Psi_{G,r} \left(g_k,g_l\right)+\frac{1}{2} \alpha_k\Vert\sigma_k \Vert^2_{\mathbb{I}.}
 	\end{equation}
 	Taking the time derivative of $V_2$
 	\begin{align}
 		&\mathcal{L}_X V_2 = K_p\left\langle \left(\mathrm{T}_e \mathscr{L}_{\gamma_k }\right)^{*} \left(\operatorname{d}_1 \Psi_{G,r} \left(\gamma_k,\gamma_l\right)\right) ;\sigma_k\right.\nonumber \\
 		&\left.- K_p\left(\mathrm{T}_e \mathscr{L}_{\gamma_k }\right)^{*} \left(\operatorname{d}_1 \Psi_{G,r} \left(\gamma_k,\gamma_l\right)\right)\right\rangle+\alpha_k \left\langle \left\langle \sigma_k,\sigma'_k\right\rangle\right\rangle\\
 		&=-K_p\Vert\left(\mathrm{T}_e \mathscr{L}_{\gamma_k }\right)^{*} \left(\operatorname{d}_1 \Psi_{G,r} \left(\gamma_k,\gamma_l\right)\right) \Vert^2_{\mathbb{I}}-R_{\operatorname{diss}} \left(\sigma_k,\sigma_k\right) .
 	\end{align}
 	Since, the initial conditions lie in the set $\mathcal{S}^2$, $\mathcal{L}_X V_2<0$ and it can be concluded that $\gamma_k \left(t\right) \mapsto \gamma_l \left(t\right)$ as $t \mapsto \infty$. Similarly, it can be shown that trajectories of all leaves connected to agent $l$ converges. Following the same procedure as in Theorem \ref{Th:Consensus_Lie_group_Obj_2}, by removing all the leaves and forming a new subgraph $\tilde{\mathcal{G}}$, it can be concluded that $\gamma_i \left(t\right) \mapsto \gamma_j \left(t\right)$ and $\gamma'_i \left(t\right) \mapsto \gamma'_j \left(t\right)$ as $t \mapsto \infty$. Thus, the desired equilibrium $\mathcal{E}^{*,2}_G$ is locally asymptotically stable.
 \end{proof}
 \subsection{Proof of Theorem \ref{Th:Consensus_Lie_group_Obj_III}}{\label{App:Proof_Th3}}
 \begin{proof}
 	Let agent $k$ be a leaf connected to agent $l$. The closed-loop dynamics of $k-$ th agent is given by
 	\begin{equation}
 		\sigma'_k=-\frac{1}{\alpha_k} \mathbb{I}^{\sharp} \left(K_p \left(\mathrm{T}_e \mathscr{L}_{\gamma_k }\right)^{*} \left(\operatorname{d}_1 \Psi \left(\gamma_k,\gamma_l\right)\right)+R^{\flat}_{\operatorname{diss}}\left(\sigma_k\right) \right).
 	\end{equation}
 	By following the procedure followed in Theorem \ref{App:Proof_Th2}, it can be proved that $\gamma_k \left(t\right) \mapsto \gamma_l \left(t\right) $ and $\gamma'_k \left(t\right) \mapsto \gamma'_l \left(t\right)$ as $t \mapsto \infty$ and can be concluded that $\gamma_i \left(t\right) \mapsto \gamma_r \left(t\right)$ and $\gamma'_i \left(t\right) \mapsto \gamma'_r \left(t\right)$ as $t \mapsto \infty$. Thus, the desired equilibrium $\mathcal{E}^{*,3}_G$ is locally asymptotically stable.
 \end{proof}                            
\bibliographystyle{IEEEtran}
\bibliography{IEEEabrv,References}
\end{document}